\documentclass[letterpaper, 11pt]{article}
\usepackage{epsf}
\usepackage{bm,mathrsfs}  
\usepackage{amsmath}
\usepackage{amssymb}
\usepackage{amsfonts}
\usepackage{mathrsfs}
\usepackage{bbm}
\usepackage{color}
\usepackage{graphicx}
\usepackage{amscd}
\usepackage{epsfig}

\definecolor{Red}{rgb}{1,0,0}
\definecolor{Blue}{rgb}{0,0,1}
\definecolor{Olive}{rgb}{0.41,0.55,0.13}
\definecolor{Green}{rgb}{0,1,0}
\definecolor{MGreen}{rgb}{0,0.8,0}
\definecolor{DGreen}{rgb}{0,0.55,0}
\definecolor{Yellow}{rgb}{1,1,0}
\definecolor{Cyan}{rgb}{0,1,1}
\definecolor{Magenta}{rgb}{1,0,1}
\definecolor{Orange}{rgb}{1,.5,0}
\definecolor{Violet}{rgb}{.5,0,.5}
\definecolor{Purple}{rgb}{.75,0,.25}
\definecolor{Brown}{rgb}{.75,.5,.25}
\definecolor{Grey}{rgb}{.5,.5,.5}
\definecolor{Black}{rgb}{0,0,0}

\newcommand{\acal}{\mathcal{A}}

\newcommand{\rcal}{\mathcal{R}}

\newcommand{\tcal}{\mathcal{T}}

\newcommand{\eps}{\varepsilon}

\newcommand{\ind}{\mathbbm{1}}

\newcommand{\bdm}{\begin{displaymath}}
\newcommand{\edm}{\end{displaymath}}
\newcommand{\bea}{\begin{eqnarray*}}
\newcommand{\eea}{\end{eqnarray*}}
\newcommand{\bean}{\begin{eqnarray}}
\newcommand{\eean}{\end{eqnarray}}

\newcommand{\prob}{\mathbb{P}}
\newcommand{\expec}{\mathbb{E}}

\newcommand{\poly}{\mathrm{poly}}

\newtheorem{theorem}{Theorem}
\newtheorem{proposition}{Proposition}
\newtheorem{corollary}{Corollary}
\newtheorem{definition}{Definition}

\newtheorem{lemma}{Lemma}

\newenvironment{proof}{\noindent{\textbf{Proof:}}}{$\blacksquare$\vskip\belowdisplayskip}
\newcommand{\qed}{$\blacksquare$\vskip\belowdisplayskip}
\newcommand{\itemname}[1]{$\mathrm{[#1]}$}

\newcommand{\weight}{\lambda}
\newcommand{\path}{\mathrm{P}}
\newcommand{\distd}{\hat d}
\newcommand{\dist}{d}
\newcommand{\pathdepth}{\Delta_{\mathrm{c}}}
\newcommand{\vertexdepth}{\Delta_{\mathrm{v}}}

\newcommand{\vpath}{\widetilde{\mathrm{P}}}
\newcommand{\cluster}[1]{\widehat{H}_{#1}}
\newcommand{\vcluster}[1]{\widehat{V}_{#1}}
\newcommand{\ecluster}[1]{\widehat{E}_{#1}}
\newcommand{\clusters}[2]{\hat{h}_{#1}^{(#2)}}
\newcommand{\vclusters}[2]{\hat{v}_{#1}^{(#2)}}
\newcommand{\eclusters}[2]{\hat{e}_{#1}^{(#2)}}
\newcommand{\ball}[2]{\widehat{B}}
\newcommand{\bipart}{b}

\newcommand{\esttree}[1]{\widehat{T}^{(#1)}}
\newcommand{\intersect}{\widehat{\Phi}}
\newcommand{\minicontractor}{\textsc{Mini Contractor}}
\newcommand{\extender}{\textsc{Extender}}
\newcommand{\clust}[1]{{H}_{#1}}
\newcommand{\vclust}[1]{{V}_{#1}}
\newcommand{\eclust}[1]{{E}_{#1}}
\newcommand{\truetree}[1]{{T}^{(#1)}}
\newcommand{\trueleaves}[1]{{L}^{(#1)}}
\newcommand{\trueintersect}{{\Phi}}
\newcommand{\trueedges}[1]{{E}^{(#1)}}
\newcommand{\truevertices}[1]{{V}^{(#1)}}

\newcommand{\minipart}[3]{\psi_{#1}(#2,#3)}
\newcommand{\fullpart}[3]{\bar\psi_{#1}(#2,#3)}
\newcommand{\minipartone}[3]{\psi^{(#2)}_{#1}(#2,#3)}
\newcommand{\miniparttwo}[3]{\psi^{(#3)}_{#1}(#2,#3)}
\newcommand{\fullpartone}[3]{\bar\psi^{(#2)}_{#1}(#2,#3)}
\newcommand{\fullparttwo}[3]{\bar\psi^{(#3)}_{#1}(#2,#3)}
\newcommand{\partsize}[2]{r(#1,#2)}

\newcommand{\radm}{m}

\newcommand{\bprime}{\beta'}

\begin{document}

\title{Phylogenies without Branch Bounds: 
Contracting the Short, Pruning the Deep}
\author{
Constantinos Daskalakis
\and
Elchanan Mossel\thanks{
E.M. is supported by an Alfred Sloan fellowship in
Mathematics and by NSF grants DMS-0528488, and DMS-0548249 (CAREER) and by 
ONR grant N0014-07-1-05-06.
}
\and
Sebastien Roch
}
\maketitle
\thispagestyle{empty}
\begin{abstract}
We introduce a new phylogenetic reconstruction algorithm
which, unlike most previous rigorous inference techniques, does not rely on assumptions regarding
the branch lengths or the depth of the tree. The
algorithm returns a forest which
is guaranteed to contain all edges that are: 1) sufficiently
long and 2) sufficiently close to the leaves. 
How much of the true tree is recovered depends on the
sequence length provided.
The algorithm is distance-based and runs in polynomial time.
\end{abstract}


\section{Introduction}\label{section:introduction}

In Evolutionary Biology, the speciation history of a family of related
organisms is generally represented graphically by a {\em phylogeny},
that is, a tree where the leaves are the observed (extant) species and
the branchings indicate speciation events.
Traditional approaches  
for reconstructing phylogenies from homologous molecular sequences
extracted from the observed species~\cite{Felsenstein:04,SempleSteel:03}
are typically 
computationally intractable~\cite{GrahamFoulds:82,DaySankoff:86,Day:87,ChorTuller:06,Roch:06}, 
statistically inconsistent~\cite{Felsenstein:78}, or they require 
impractical sequence lengths \cite{Atteson:99,LaceyChang:06,SteelSzekely:99,SteelSzekely:02}.
Nevertheless, over the past decade, 
much progress has been made in the design of efficient, fast-converging reconstruction techniques, 
starting with the seminal
work of Erd\"os et al.~\cite{ErStSzWa:99a}. The algorithm in~\cite{ErStSzWa:99a}, often dubbed 
the Short Quartet Method (SQM), is based on well-known distance-matrix techniques, that is,
it relies on estimates of the evolutionary
distance between each pair of species (roughly the time elapsed since their
most recent common ancestor).
However, 
unlike other popular distance methods such as Neighbor-Joining~\cite{SaitouNei:87}, 
the key behind SQM's performance is 
that it discards long evolutionary distances, 
whose estimates from sequence comparisons are known to be statistically unreliable. 
The algorithm works by first building subtrees 
of small diameter and, in a second stage, glueing the pieces back together. 

The Short Quartet Method is in fact guaranteed to return the correct topology from polynomial-length sequences
in polynomial time with high probability. 
But this appealing theoretical performance comes at a price.
The results of~\cite{ErStSzWa:99a} rely critically on biological assumptions which, although reasonable, 
are often not met in practice (see Section~\ref{sec:related} for a formal statement):
\begin{itemize}
\item[a)] {\it [Dense Sampling of Species]}
The observed species are ``closely related.'' In particular, there are no exceptionally long
branches in the phylogeny. 

\item[b)] {\it [Absence of Polytomies]} 
The phylogeny is bifurcating. In fact, Erd\"os et al.~assume that speciation 
events are sufficiently far apart to be easily distinguished.

\end{itemize} 
The point of a) is that it implies a natural bound on the depth of the tree which in turn ensures 
that enough information about the deep parts of the tree diffuses to the leaves. 
As for Assumption b), it guarantees that a clear signal 
can be extracted from each branch of the phylogeny.
It is obvious---at least intuitively---that assumptions such as a) and b) are
necessary to secure the type of results Erd\"os et al.~obtain: 
\emph{the guaranteed reconstruction of the full phylogeny}. 
Hence, to improve over SQM and obtain strong guarantees under more general
conditions, one has to relax this last requirement.

In this paper, we design an algorithm which provides strong reconstruction guarantees without Assumptions a) and b). 
We show that our algorithm is guaranteed to recover a {\em forest} containing all edges that are ``sufficiently long'' and ``sufficiently close'' to the leaves.  
In fact, we allow a
trade-off between the resolution of short branches and the depth of the reconstructed forest, a feature of potential practical interest. 
Also, we guarantee that our reconstructed forest has the desirable property of being \emph{disjoint} 
(although the presence of short edges leads us to allow deep intersections of very short branches between the subtrees). 
Moreover, our algorithm does not require the knowledge of a priori bounds on branch lengths or tree depth. 
Finally if Assumptions a) and b) are satisfied, we recover the whole phylogeny and provide an alternative to the algorithm of Erd\"os et al. 

Precise statements are given in Section~\ref{section:main}. For a full comparison to related work see also Section~\ref{sec:related}.

\subsection{What can we hope to reconstruct?}\label{section:definitions}


Well-known identifiability results~\cite{Chang:96} guarantee that phylogenies---or at least their idealized stochastic models---can be fully reconstructed
given enough data at the leaves. However, molecular data gathered from current species are in essence limited, which begs the question: 
{\em How much of the tree can we really hope to reconstruct?} We pointed out above two important sources of difficulties: short branches 
produce a weak signal that may be hard to detect; similarly, untangling the deep parts of the tree presents challenges 
that are well documented (see, e.g.,~\cite{PhilippeLaurent:98,CDvMCSB:06}).
Note that these issues are fundamentally ``information-theoretic'' and affect all reconstruction methods.


To avoid these difficulties, most {\em rigorous} methods impose restrictions on the length of the branches and/or the depth of the tree, 
which may be unsatisfactory from a practical perspective. 
On the other hand, methods commonly used in {\em practice}, such as likelihood and bayesian methods, 
typically produce several candidate trees as well as confidence estimates. 
But theoretical guarantees on the quality of such outputs are hard to obtain. 

Here, we seek to give strong reconstruction guarantees without any assumption on the true phylogeny. 
Our goal is to recover, for any given amount of data, as much of the tree as can rigorously be reconstructed with high confidence. 
Since the full phylogeny may not always be recoverable, 
we are led to a more flexible solution concept: 
we output a {\em contracted subforest} of the true phylogeny. 
That is, we output a forest containing all branches that are ``sufficiently long'' and ``sufficiently recent''; 
note that ``sufficiently'' here is determined (information-theoretically) 
by the size of the data (usually in terms of sequence length). 
In the remainder of this section we formalize this solution concept.

\paragraph{The input.} Formally, a phylogeny is a \emph{weighted, multifurcating tree} on a set of leaves $L$, which we identify with the labels $[n] = \{1,\ldots,n\}$. 
We denote a phylogeny by $T = (V,E;L,\weight)$. Here $V$ and $E$ are respectively the vertex and edge set of the tree, 
and $\weight : E \to (0,+\infty)$ assigns a weight to each edge (the branch length). 
We assume that all internal vertices $V-L$ have degree at least $3$.

A phylogeny is naturally equipped with a so-called additive metric on the leaves
$\dist : L\times L \to (0,+\infty)$ defined as follows
\begin{equation*}
\forall u,v \in L,\ \dist(u,v) = \sum_{e\in\path_T(u,v)} \weight_e,
\end{equation*}
where $\path_T(u,v)$ is the set of edges on the path between
$u$ and $v$ in $T$. Often $\dist(u,v)$ is referred to as the ``evolutionary distance'' between species $u$ and $v$. Since under the assumptions above there is a one-to-one
correspondence between $\dist$ and $\weight$, 
we write either $T = (V,E;L,\dist)$ or $T = (V,E;L,\weight)$.
We also sometimes use the natural extension of $\dist$ to the internal vertices
of $T$. We denote by $\tcal$ the set of all phylogenies on any number
of leaves.

It is well-known that given an additive metric $\dist$ one can reconstruct the corresponding phylogeny $T$. 
However, in practice, one can only derive an {\em estimate} $\distd$ of $\dist$, 
the accuracy of which depends on the amount of data used. 
(This estimate is known in the literature as the ``distance matrix''.) 
Our goal in this paper is to reconstruct a phylogeny---or as much of it as possible---from 
this ``distorted'' version of its additive metric. A typical property of distance estimates 
is that estimates of long distances are unreliable. 
The following definition formalizes this phenomenon.
See Figure~\ref{fig:distorted} for an illustration.
\begin{figure}
\centering
\includegraphics[width=9cm]{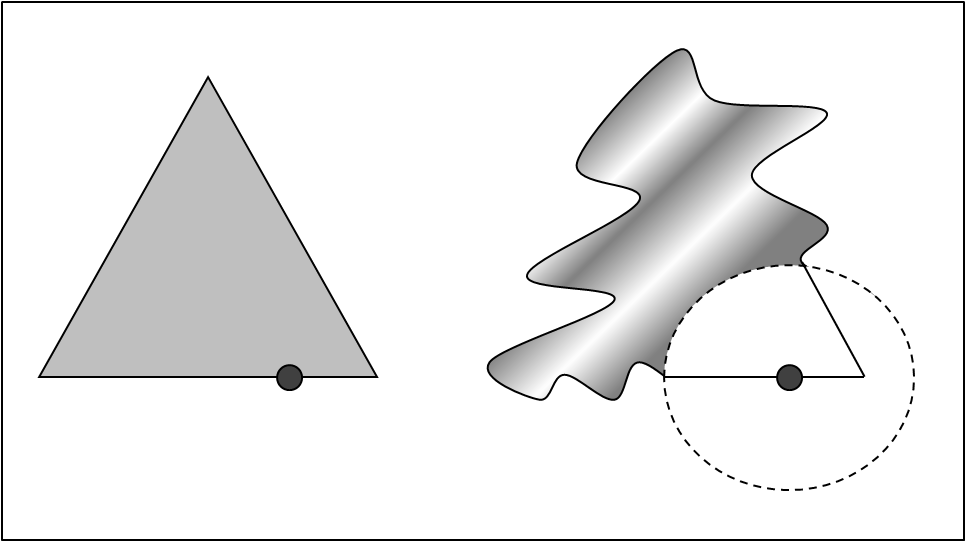}
\caption{The effect of distance distortion from the perspective of a leaf. On the left hand side is the true phylogeny.
On the right hand side, only distances within a certain radius represent accurately the metric underlying the phylogeny.}
\label{fig:distorted}
\end{figure}
\begin{definition}[Distorted Metric~\cite{Mossel:07,KiZhZh:03}]\label{def:distorted metric}
Let $T = (V,E;L,\dist)$ be a phylogeny and let $\tau, M > 0$. 
We say that $\distd : L\times L \to (0,+\infty]$ is a $(\tau, M)$-\emph{distorted metric}
for $T$ or a $(\tau, M)$-\emph{distortion} of $\dist$ if:
\begin{enumerate}
\item \itemname{Symmetry}
For all $u,v \in L$, $\distd$ is symmetric, that is,
\begin{equation*}
\distd(u,v) = \distd(v,u);
\end{equation*}
\item \itemname{Distortion} $\distd$ is accurate on
``short'' distances, that is, for all $u,v \in L$, if either 
$\dist(u,v) < M + \tau$ or $\distd(u,v) < M + \tau$ then
\begin{equation*}
\left|\dist(u,v) - \distd(u,v)\right| < \tau.
\end{equation*}
\end{enumerate}
\end{definition}
In phylogenetic reconstruction, a distorted metric is naturally derived from samples of a Markov model on a tree---a common model of DNA sequence evolution used in Biology. 
(See Appendix~\ref{section:markov} for details.) 
In the remainder of this paper, we assume that we are given a $(\tau,M)$-distortion $\distd$ of an additive metric $\dist$ 
and we seek to recover the underlying phylogeny $T$.

\paragraph{Contraction and pruning.} 
Given only a $(\tau,M)$-distorted metric, it is clear that the best we can hope for in general is to reconstruct a forest containing those edges of $T$ that are ``sufficiently close'' to the leaves. 
Indeed, note that two phylogenies that are identical up to depth $M$ from the leaves, but are otherwise different, can give rise to the same distorted metric. 
Moreover, since we do not assume that edges are longer than the accuracy $\tau$, some edges may be too short to be reconstructed 
and, as we mentioned before, we allow ourselves to instead contract them. 
Hence, we are led to consider subforests of the true phylogeny 
where deep edges are {\em pruned} and short edges are {\em contracted}.

To formalize this idea we need a few definitions. Let us first describe what we mean by a {\em subforest} of a phylogeny $T = (V,E;L,\dist)$. 
Given a set of vertices $V' \subseteq V$, the {\em subtree of $T$ restricted to $V'$} is the tree obtained 1) by keeping only nodes and edges on paths between vertices in $V'$ and then 
2) by contracting all paths composed of vertices of degree 2, except the nodes in $V'$. 
See Figure~\ref{fig:restricted} for an example. We denote this tree by $T|_{V'}$. 
We typically take $V' \subseteq L$. A {\em subforest} of $T$ is defined to be a collection of restricted subtrees of $T$.
\begin{figure}
\begin{center}
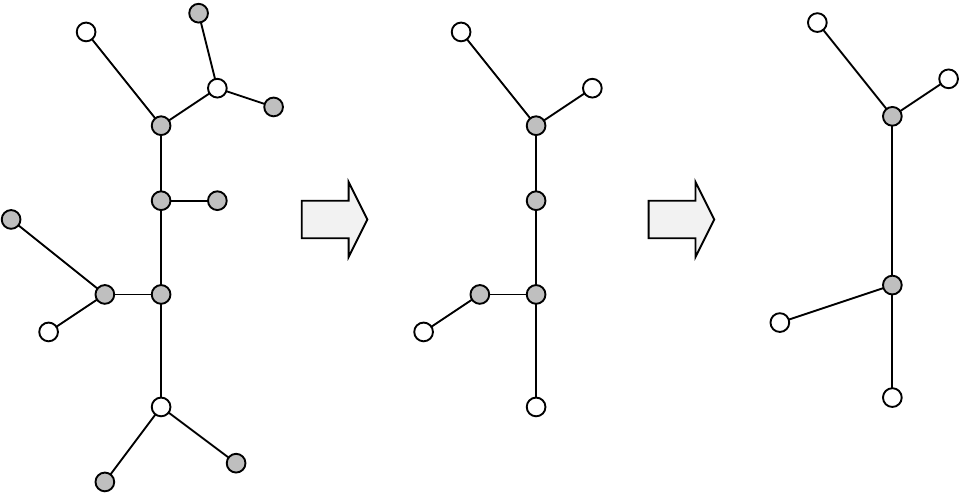\caption{Restricting the top tree to its white nodes.}\label{fig:restricted}
\end{center}
\end{figure}

We also need a notion of depth. 
Given an edge $e \in E$, the \emph{chord depth} of $e$ 
is the length of the shortest path among all paths crossing $e$ between two leaves. 
That is,
\begin{equation*}
\pathdepth(e) = \min\left\{\dist(u,v)\ :\ u,v \in L, e\in \path_T(u,v)\right\}. 
\end{equation*}
We define the \emph{chord depth of a tree $T$} to be the maximum
chord depth in $T$
\begin{equation*}
\pathdepth(T) = \max\left\{\pathdepth(e)\ :\ e\in E\right\}.
\end{equation*}

\begin{definition}[Contracted Subforest]
Let $T = (V,E;L,\dist)$ be a phylogeny. Fix $M > 0$. Let $\{L_1, \ldots, L_q\}$ be the natural partition of the leaf set $L$ obtained by removing all edges $e\in E$ such that $\pathdepth(e) \geq M$. We define the $M$-\emph{pruned subforest} of $T$ to be the forest $F_M(T) = (V_M, E_M)$ consisting of the trees $\{T|_{L_1},\ldots,T|_{L_q} \}$. The metric $\dist$ is extended as follows for all $u,v \in L$,
\begin{displaymath}
\dist_M(u,v) =
\left\{ 
\begin{array}{ll}
\dist(u,v), & \mathrm{if\ }u,v\mathrm{\ are\ in\ the\ same\ subtree\ of\ }F_M(T),\\
+\infty, & \mathrm{o.w.}
\end{array}
\right.
\end{displaymath}
We also denote by $\weight_M$ the edge lengths of $F_M(T)$.

Now, given also $\tau > 0$, the $\tau$-\emph{contracted} $M$-\emph{pruned subforest}
of $T$ is the forest $F_{\tau, M}(T) = (V_{\tau,M}, E_{\tau,M})$ obtained from $F_M(T)$ 
by contracting edges $e\in E_M$ of weight $\weight_M(e) \leq \tau$.
\end{definition}

\paragraph{Path-disjointness.} We require that
the trees of our reconstructed forest are ``non-intersecting''. 
This is a natural condition to impose in order to obtain a
meaningful reconstruction: we want to avoid as much as possible that the same
branches appear in many subtrees.
In fact, we can only guarantee
approximate disjointness as defined below.

We first need a notion of depth for vertices. For a phylogeny $T = (V,E;L,\dist)$ and a vertex $x \in V$, 
the \emph{vertex depth} of $x$ is the length of the shortest path between $x$ and the set of leaves. That is,
\begin{equation*}
\vertexdepth(x) = \min\left\{\dist(u,x)\ :\ u \in L\right\}. 
\end{equation*}
Given two leaves $u,v$ of $T$, we denote by $\vpath_{T}(u,v)$
the set of vertices on the path between $u$ and $v$ in $T$.

We say that two trees are $(\tau, M)$-path disjoint if they are ``almost disjoint'' 
in the sense that they only share edges (if any) that are ``deep'' (endpoints have vertex depth at least $M/2$) 
and ``short'' (length at most $\tau$). More formally:
\begin{definition}[Approximate Path-Disjointness] \label{def:paht disjointness}
Let $T = (V,E;L,\dist)$ be a phylogeny. 
Two subtrees $T_1, T_2$ of $T$ restricted respectively to
$L_1, L_2 \subseteq L$ are \emph{$(\tau,M)$-path-disjoint} 
if $L_1 \cap L_2 = \emptyset$ and for all pairs of leaves $u_1,v_1 \in L_1$ and $u_2,v_2 \in L_2$ such that
\begin{equation*}
\vpath_T(u_1,v_1)\cap \vpath_T(u_2,v_2) \neq \emptyset,
\end{equation*}
we have:
\begin{equation*}
\min\{\vertexdepth(x)\ :\ x \in \vpath_T(u_1,v_1)\cap \vpath_T(u_2,v_2)\} \geq \frac{1}{2}M,
\end{equation*}
and, if further $\path_T(u_1,v_1)\cap \path_T(u_2,v_2) \neq \emptyset$,
\begin{equation*}
\max\{\weight_e\ :\ e \in \path_T(u_1,v_1)\cap \path_T(u_2,v_2)\} \leq \tau.
\end{equation*}
More generally, a collection of restricted subtrees $T_1,\ldots,T_q$ of $T$ are \emph{$(\tau,M)$-path-disjoint}
if they are pairwise $(\tau,M)$-path-disjoint. In the case $\tau = 0$, we simply say that
the subtrees are \emph{path-disjoint}.
\end{definition}

\subsection{Main result and corollaries}\label{section:main}

\paragraph{Main result.} Our main result is an algorithm which, given a $(\tau,M)$-distorted metric, 
reconstructs a contracted subforest (of the true phylogeny) whose trees are approximately path-disjoint.
Typically, $M$ is much larger than $\tau$.
In that case, we reconstruct a subforest of $T$ with chord depth $\approx \frac{1}{2} M$ which includes
all edges of length at least $4\tau$. The reconstructed subtrees may ``overlap'' on edges of length at most $2\tau$ at
vertex depth $\gtrapprox \frac{1}{4} M$. In Section~\ref{section:lower}, we show that these
parameters are essentially optimal. The algorithm runs in polynomial time. 

More precisely, we show:
\begin{theorem}[Main Result]\label{theorem:main}
Let $\tau$ and $M$ be monotone functions of $n$ with $M > 3\tau$. 
Let $\radm > 3\tau$ be such that
\begin{equation*}
\radm < \frac{1}{2}[M - 3\tau],
\end{equation*}
for all $n$.
Then, 
there is a polynomial-time algorithm $\acal$ 
such that, for all phylogenies $T = (V,E;L,\dist)$ in $\tcal$ with $|L| = n$
and all $(\tau,M)$-distortions $\distd$ of $\dist$,
$\acal$ applied to $\distd$ satisfies the following:
\begin{enumerate}
\item \itemname{Approximate\ Path\ Disjointness}\label{item:main1} 
$\acal$ returns a $(2\tau, \radm - 3\tau)$-path-disjoint subforest $\widehat F$ of $T$;
\item \itemname{Depth\ Guarantee}\label{item:main2} 
The forest $\widehat F$ is a refinement of $F_{4 \tau, \radm - \tau}(T)$;
\end{enumerate}
\end{theorem}
We give below a few important special cases of Theorem~\ref{theorem:main}.

\paragraph{Tree case.}
When the amount of data is sufficient to produce a distorted metric with $M = \Omega(\pathdepth(T))$, 
we get a single component, that is, the full tree (up to those edges that are contracted).
\begin{corollary}[Tree Case]\label{cor:tree case}
Let $\tau > 0$ and $M > 2\pathdepth(T) + 5\tau$. Then,
choosing $\radm > \pathdepth(T) + \tau$ guarantees that the reconstructed forest
is composed of only a tree.
\end{corollary}
In the case of ``dense'' phylogenies, $M = \Omega(\log n)$ is sufficient
to reconstruct the full tree.
\begin{definition}[Dense Phylogenies (see e.g.~\cite{ErStSzWa:99a})]
We say that a collection of phylogenies $\tcal'$ is \emph{dense} if there is
a $0 < g < +\infty$ (independent of $n$) such that for all $T = (V,E; L,\weight)\in \tcal'$ we have
\begin{equation}\label{eq:dense}
\forall e\in E,\ \weight_e \leq g.
\end{equation}
We denote by $\tcal_g$ the set of phylogenies satisfying (\ref{eq:dense}).
\end{definition}
\begin{corollary}[Dense Case]\label{cor:dense}
In the case of dense phylogenies, $M = \Omega(\log n)$ suffices to guarantee
the reconstruction of the full tree, up to contracted edges.
\end{corollary}

\paragraph{Absolute variant.} All rigorous algorithms prior to our work (see Section~\ref{sec:related}) 
require knowledge of either the tree depth or bounds on the edge lengths to give strong reconstruction guarantees. 
This is not satisfactory from a practical point of view. Here given only the sequence length we provide explicit guarantees. 
The following result assumes that the distorted metric is derived from a Markov model on a tree. (See Appendix~\ref{section:markov} for details.)
\begin{corollary}[Absolute Variant]\label{cor:absolute}
Given a number of samples $k = \Omega(\log n)$ from a Markov model on a tree 
and a chosen level of contraction $\eps > 0$ (small), one can choose
$\tau, M, \radm$ so that
$\acal$ is guaranteed to return a (contracted) subforest of $T$ containing 
$F_{\eps, M'}(T)$ with probability $1 - o(1)$,
where $M' = \Omega_\eps(\log k - \log\log n)$.
\end{corollary}

\paragraph{Complete resolution.} Finally we remark that, if we further assume that all branch lengths are bounded {\em from below} by a constant, 
then by choosing $\tau$ accordingly a non-contracted forest is returned.  In particular, we can recover the results of~\cite{ErStSzWa:99a}.

\subsection{Related work} \label{sec:related}

Under a Markov model of evolution, 
the Short Quartet Method (SQM) of Erd\"os et al.~\cite{ErStSzWa:99a} is guaranteed 
to recover the full phylogeny as long as the number of samples $k$ satisfies 
\begin{equation*}
k > c f^{-2} e^{c' g \Delta_c(T)} \log n,
\end{equation*}
for constants $c,c' > 0$, where
$f$ and $g$ are respectively lower and upper bounds on the branch lengths 
possibly depending on $n$.
For instance, if $f$ and $g$ are constants 
the sequence length needed for complete reconstruction depends polynomially in the number of species.


Mossel~\cite{Mossel:07} developed a framework that allows the reconstruction of a well-behaved \emph{forest} when sequences are too short to guarantee a complete reconstruction. 
More precisely, edges which are too deep (in the sense of appearing only on paths between species whose distances are not accurately known) are \emph{pruned} from the final reconstruction. 
At a high level, Mossel's Distorted Metric Method (DMM) (implicit in~\cite{Mossel:07}), 
works in a fashion similar to SQM---except for a pre-processing phase that clusters together sufficiently related species. 
However, for DMM to work, lower bounds on the branch lengths are required and, moreover, these must be known by the algorithm. 
Following up on~\cite{Mossel:07}, Daskalakis et al.~\cite{DHJMMR:06} gave a variant of DMM that runs without 
knowledge of a priori bounds on the branch lengths or the tree depth---making their variant somewhat more practical. 
However, like DMM, the algorithm in~\cite{DHJMMR:06} does not deal properly with short edges: any part of the tree containing 
a short edge cannot be reconstructed by the algorithm (even though there may be adjacent edges that are in fact reconstructible). 
Therefore, in the presence of short edges no guarantee can be given about the depth of the reconstructed forest.

Recently Gronau et al.~\cite{GrMoSn:08} eliminated the need for a lower bound on the branch length by \emph{contracting} edges whose length is below a user-defined threshold. 
Their solution uses a Directional Oracle (DO) which closes in on the location of a leaf to be added and, 
in the process, contracts regions that do not provide a reliable directional signal. 
Although the DO algorithm does not use an explicit bound on the depth of the tree, 
their {\em reconstruction guarantee} requires such a bound, similarly to~\cite{ErStSzWa:99a}. 
In particular, Gronau et al.~leave open the question of giving a forest-building version
of their algorithm. 
Moreover, the sequence length in~\cite{GrMoSn:08} depends exponentially on what the authors 
call the $\eps$-diameter of the tree---essentially, the maximum diameter of the contracted regions.  
It is natural to conjecture that an optimal result should not depend on this parameter.

For further related work on efficient phylogeny reconstruction, 
see also~\cite{ErStSzWa:99b,HuNeWa:99,CsurosKao:01,Csuros:02,KiZhZh:03,MosselRoch:06,DaMoRo:06}. 


\subsection{Discussion of the results}

In Table~\ref{fig:comparison} we summarize the current status as discussed in the previous sections.  
 
\begin{figure}
\centering
\begin{tabular}{|c||c|c|c|}
\hline
 &\multicolumn{1}{c|}{\it Execution}&\multicolumn{2}{c|}{\em Guarantees}\\
\cline{2-4}
 &No branch&Short edges&Deep edges\\
 & bound needed&OK&OK\\

\hline \hline
\cite{ErStSzWa:99a}&~&~&~\\
\hline
\cite{Mossel:07}&~&~&$\checkmark$\\
\hline
\cite{DHJMMR:06}&$\checkmark$&~&$\checkmark$\\
\hline
\cite{GrMoSn:08}&$\checkmark$&$\checkmark$&~\\
\hline
{\bf Our method}&\bf $\checkmark$&\bf $\checkmark$&\bf $\checkmark$\\
\hline
\end{tabular}
\caption{Comparison of methods.}\label{fig:comparison}
\end{figure}
As the table emphasizes, our overarching goal is to design an algorithm with good reconstruction guarantees
in the presence of both short and deep edges, whose execution does not rely on a priori bounds on branch lengths.
Unfortunately, given the combinatorial complexity of Mossel's forest-building algorithm, 
it is not straightforward to provide 
the extra flexibility of edge contraction in this framework. 
The novelty in our work is twofold:
\begin{itemize}
\item {\em Solution Concept:} A basic complication is that, in some sense, 
contraction and pruning interfere with each other. 
Indeed, the presence of unresolved branches at the boundary of partially reconstructed subtrees 
creates the possibility of deep ``undetectable'' intersections. 
This pitfall seems to be unavoidable.
One of our main contributions is to introduce the notion of approximate disjointness, 
which allows short but deep intersections between subtrees of the reconstructed forest. 
This suitable solution concept leads to a quite simple algorithm with reasonable guarantees. 
Moreover, the flexibility in our definition allows us to recover all previously known results as special cases.

\item {\em Algorithmic Technique:} A natural approach to forest building used in~\cite{Mossel:07,DHJMMR:06} proceeds along the following three steps: 
\begin{enumerate}
\item first, leaves are grouped into clusters for which all pairwise distances are accurately known (the {\em small} clusters);

\item by definition, the local topologies on the small clusters can be trivially reconstructed~\cite{Buneman:71};

\item finally, the local topologies that intersect in the true tree are ``glued'' together to get a forest (the resulting forest partitions the leaves into {\em large} clusters).

\end{enumerate}
This last step involves non-trivial combinatorial considerations. We have found that further allowing contracted edges makes this process somewhat unmanageable. Instead we use a different approach relying on simple metric arguments. In particular, we {\em directly} partition the leaves into large clusters, whose underlying subtrees are approximately disjoint, and provide a new straightforward method to reconstruct these subtrees. 
\end{itemize}

In addition, we obtain as special cases the results discussed in Section~\ref{sec:related}. 
In particular, if there are no short edges, we recover the results of~\cite{Mossel:07} and~\cite{DHJMMR:06}, 
where a path-disjoint forest is returned (by taking $\tau$ equal to half the lower bound on the branch lengths in Theorem~\ref{theorem:main}). 
If furthermore there is an upper bound on the branch lengths, we recover the results of~\cite{ErStSzWa:99a} (Corollary~\ref{cor:dense}). 
Finally, if we keep the upper bound on the edge lengths, but drop the lower bound, we recover the results of~\cite{GrMoSn:08} (Corollary~\ref{cor:tree case}). 
In fact, we eliminate the dependence on the $\eps$-diameter.~\footnote{After the results of the current paper were posted on the arXiv, 
we were informed by S.~Moran that, in parallel to our work, the authors 
of~\cite{GrMoSn:08} have improved on their previous results: 
the dependence on the $\eps$-diameter has been removed. 
A preprint of this work is currently 
available on the authors' website.
Note however that this new, independent work does not deal with deep 
edges and still makes assumptions similar to~\cite{ErStSzWa:99a} restricting 
the depth of the generating tree.
} 
Further, unlike~\cite{GrMoSn:08}, we allow an arbitrary number of states, an extension---it should be noted---that follows easily from~\cite{ErStSzWa:99b} and~\cite{Mossel:07}.


\subsection{Organization}

The rest of the paper is organized as follows. 
The algorithm is detailed in Section~\ref{section:algorithm}.
The proof of our main theorem follows in Section~\ref{section:proof}.
We conclude with a lower bound in Section~\ref{section:lower}
and a discussion of the running time in Section~\ref{section:running}.
Also, for completeness, in Appendix~\ref{section:markov} 
we describe the probabilistic motivation behind the distorted metric definition. 

The results in this paper were announced without proof in~\cite{DaMoRo:09}.
Also, the counter-example in Section~\ref{section:lower} did not appear in~\cite{DaMoRo:09}.

\section{Algorithm}\label{section:algorithm}

The outline of the algorithm follows. There are three main phases, which are explained in detail after the outline. The input to the algorithm is a $(\tau, M)$-distorted
metric $\distd$ on $n$ leaves. In particular, we assume that the values $\tau$ and
$M$ are known to the algorithm (but see also Corollary~\ref{cor:absolute}). 
Let $\radm$ be as in Theorem~\ref{theorem:main}.
We denote the true tree by $T = (V,E;L,\dist)$. 
The details of the subroutines \minicontractor~and \extender~are detailed in Figures~\ref{figure:minicontractor} and~\ref{figure:extender} (see also their high level description below). 

\begin{itemize}
\item \textbf{Pre-Processing: Leaf Clustering.} 
Build the distorted clustering graph
$\cluster{\radm} = (\vcluster{\radm}, \ecluster{\radm})$ where $\vcluster{\radm} = [n]$
and
$(u,v) \in \ecluster{\radm} \iff \distd(u,v) < \radm$;
compute the connected components 
$\{\clusters{\radm}{i} = (\vclusters{\radm}{i},\eclusters{\radm}{i})\}_{i=1}^q$
of $\cluster{\radm}$;
\item \textbf{Main Loop.} 
For all components $i=1,\ldots,q$: 
\begin{itemize}
\item For all pairs of leaves
$u,v \in \vclusters{\radm}{i}$ such that $(u,v) \in \ecluster{\radm}$:  
\begin{itemize}
\item \textbf{Mini Reconstruction.}
Compute 
\begin{equation*}
\{\minipart{j}{u}{v}\}_{j=1}^{\partsize{u}{v}} := \mathrm{\minicontractor}(\clusters{\radm}{i}; u,v);
\end{equation*}
\item \textbf{Bipartition Extension.}
Compute 
\begin{equation*}
\{\fullpart{j}{u}{v}\}_{j=1}^{\partsize{u}{v}} := \mathrm{\extender}(\clusters{\radm}{i}, \{\minipart{j}{u}{v}\}_{j=1}^{\partsize{u}{v}}; u,v);
\end{equation*}
\end{itemize}
\item Deduce the tree $\esttree{i}$ from $\{\fullpart{j}{u}{v}\}_{j=1}^{\partsize{u}{v}}$;
\end{itemize}
\item \textbf{Output.} Return the resulting forest $\widehat{F}$.
\end{itemize}

\paragraph{Pre-processing: Leaf clustering.}
As mentioned before, given a $(\tau,M)$-distortion we cannot hope to reconstruct
edges that are too deep inside the tree. This results in the reconstruction
of a \emph{forest}. Therefore, the first phase of the algorithm is to determine the
``support'' of this forest. We proceed as follows. Consider the following graph on $L$.
\begin{definition}[Clustering Graph]\label{definition:clustering}
Let $M' \in [\tau, \leq M - \tau]$. The \emph{distorted clustering graph with parameter $M'$}, 
denoted $\cluster{M'} = (\vcluster{M'}, \ecluster{M'})$, 
is the following graph:
the vertices $\vcluster{M'}$ are the leaves $L$ of $T$; 
two leaves $u,v \in L$ 
are connected by an edge $e = (u,v) \in \ecluster{M'}$ if
\begin{equation}\label{eq:clustering}
\distd(u,v) < M'.
\end{equation}
Note that this is an undirected graph because $\distd$ is symmetric.
Similarly, we define the \emph{clustering graph with parameter $M'$},
$\clust{M'} = (\vclust{M'}, \eclust{M'})$, where we use
$\dist$ instead $\distd$ in (\ref{eq:clustering}).
\end{definition}
The first phase of the algorithm consists in building the graph $\cluster{\radm}$
from $\distd$. We then compute the connected components of $\cluster{\radm}$ which
we denote $\{\clusters{\radm}{i}\}_{i=1}^q$. In the next two phases, we build a tree
on each of these components.

\paragraph{Building the components I: Mini-reconstruction problem.} 
\begin{figure}
\begin{center}
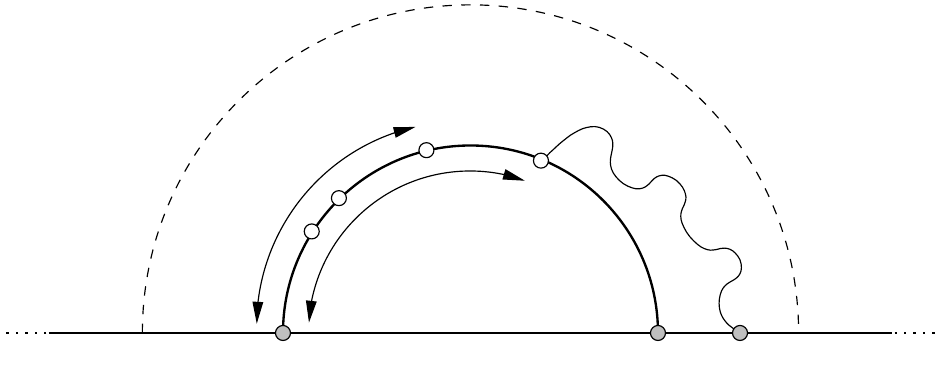\caption{Illustration of routine \minicontractor. See Figure~\ref{figure:minicontractor} for notation.}\label{fig:mini}
\end{center}
\end{figure}
\begin{figure*}
\framebox{
\begin{minipage}{11.7cm}
{\small \textbf{Algorithm} \minicontractor\\
\textit{Input:} Component $\clusters{\radm}{i}$; Leaves $u,v$;\\
\textit{Output:} Bipartitions $\{\minipart{j}{u}{v}\}_{j=1}^{\partsize{u}{v}}$;
\begin{itemize}
\item \textbf{Ball.} Let
\begin{equation*}
\ball{{\bprime}}{i}(u,v) := \left\{w \in \vclusters{\radm}{i}\ :\ 
\distd(u,w) \lor \distd(v,w) < M \right\};
\end{equation*}
\item \textbf{Intersection Points.} For all $w\in \ball{\bprime}{i}(u,v)$, estimate the point of intersection
between $u,v,w$ (distance from $u$), that is,
\begin{equation*}
\intersect_w := \frac{1}{2}\left(
\distd(u,v) + \distd(u,w) - \distd(v,w)
\right);
\end{equation*}
\item \textbf{Long Edges.} 
Set $S := \ball{\bprime}{i}(u,v) - \{u\}$, $x_{-1} = u$, $j := 0$, $C_0 = \{u\}$; 
\begin{itemize}
\item Until $S=\emptyset$:
\begin{itemize}
\item Let $x_0 = \arg\min\{\intersect_w\ :\ w \in S\}$ (break ties arbitrarily); 
\item If $\intersect_{x_0} - \intersect_{x_{-1}} \geq 2 \tau$, create a new edge
by setting $\minipart{j+1}{u}{v} := \{\ball{\bprime}{i}(u,v) - S, S\}$ and let $C_{j+1} := \{x_0\}$, $j := j+1$;
\item Else, set $C_j := C_j \cup \{x_0\}$;
\item Set $S := S-\{x_0\}$, $x_{-1} := x_0$;
\end{itemize}
\end{itemize}
\item \textbf{Output.} Return the bipartitions $\{\minipart{j}{u}{v}\}_{j=1}^{\partsize{u}{v}}$ (where $\partsize{u}{v}$ is the number
of bipartitions generated in the previous step).
\end{itemize}
}
\end{minipage}
} \caption{Algorithm \minicontractor. See Figure~\ref{fig:mini} for illustration.} \label{figure:minicontractor}
\end{figure*}
Fix a component $\clusters{\radm}{i}$ of $\cluster{\radm}$. In this and the next phase,
we seek to reconstruct a contracted tree on $\clusters{\radm}{i}$. 
Denote by
$\truetree{i}$ the true tree $T$ restricted to the leaves in $\clusters{\radm}{i}$.
First, we find all 
edges of $\truetree{i}$ that are ``sufficiently long'' and lie on ``sufficiently short'' paths. 
More precisely, we consider all pairs of leaves $u,v$ connected by an edge in 
$\clusters{\radm}{i}$, that is, leaves within distorted distance $\radm$. For each such
pair $u,v$, the \emph{mini reconstruction problem} consists in finding all edges $e$
in $\path_{\truetree{i}}(u,v)$ that have length larger than $\weight_e \geq 4\tau$.
To do this using the distortion $\distd$, we first consider a \emph{ball} $\ball{\bprime}{i}(u,v)$ of all
nodes within distorted distance $M$ of $u$ and $v$, 
that is,
\begin{equation*}
\ball{\bprime}{i}(u,v) = \left\{w \in \clusters{\radm}{i}\ :\ 
\distd(u,w) \lor \distd(v,w) < M \right\},
\end{equation*}
where $a \lor b$ is the maximum of $a$ and $b$.
The point of using this ball is that
we can then guarantee that each edge in $\path_{\truetree{i}}(u,v)$
is ``witnessed'' by a quartet (i.e., a $4$-tuple of leaves) 
in $\ball{\bprime}{i}(u,v)$ in the following sense: let
$(x_1,x_2)$ be an edge in $\path_{\truetree{i}}(u,v)$ and let $(x_j, y_j)$, $j=1,2$, 
be an edge adjacent 
to $x_j$ that is \emph{not} in $\path_{\truetree{i}}(u,v)$; for $j=1,2$ let $\trueleaves{i}_{x_j \to y_j}$ be the leaves reachable from $y_j$ using paths not including $x_j$; then we will show
that $\trueleaves{i}_{x_j \to y_j} \cap \ball{\bprime}{i}(u,v) \neq \emptyset$ for $j=1,2$. In other words,
there is enough information in $\ball{\bprime}{i}(u,v)$ to reconstruct all edges
in $\path_{\truetree{i}}(u,v)$---at least those that are ``sufficiently long.''  
This phase is detailed in Figure~\ref{figure:minicontractor}.
An illustration is given in Figure~\ref{fig:mini}.

\paragraph{Building the components II: Extending the bipartitions.}
The previous step reconstructs ``sufficiently long'' edges on balls of the form
$\ball{\bprime}{i}(u,v)$. By \emph{reconstructing an edge on $\ball{\bprime}{i}(u,v)$}, we mean
\emph{finding the bipartition of $\ball{\bprime}{i}(u,v)$ to which the edge corresponds}.
More precisely:
\begin{definition}[Bipartitions]
Let $T = (V,E)$ be a multifurcating tree with no vertex of degree $2$. 
Each edge $e$ in $T$ induces a \emph{bipartition}
of the leaves $L$ of $T$ as follows: if one removes the edge $e$ from $T$,
then one is left with two connected components; take the partition of the leaves
corresponding to those components. Denote by $\bipart_T(e)$ the bipartition of $e$
on $T$. It is easy to see that given the bipartitions $\{\bipart_T(e)\}_{e\in E}$
one can reconstruct the tree $T$ efficiently~\cite{Buneman:71,Meacham:81,BandeltDress:86}. 
(Proceed by sequentially ``splitting'' clusters.) 
\end{definition}
The goal of the second phase in the main loop of our reconstruction algorithm is to {\em extend}
the bipartitions previously built from $\ball{\bprime}{i}(u,v)$ to the full component 
$\clusters{\radm}{i}$. To perform this task, we use the following observation: suppose
we want to deduce the bipartition corresponding to edge $e$; since the radius of the
ball $\ball{\bprime}{i}(u,v)$ is much larger than $\radm$,
we can make sure that a path {\em from} a leaf in $\clusters{\radm}{i}$ that is outside 
$\ball{\bprime}{i}(u,v)$ {\em to} a leaf on the other side of the bipartition $\bipart_T(e)$ is
``long.'' Therefore, we can easily determine what side of the partition
each leaf in $\clusters{\radm}{i}$ lies on. For details, see Figure~\ref{figure:extender}.
An illustration is given in Figure~\ref{fig:extend}.
\begin{figure}
\begin{center}
\includegraphics[width=9cm]{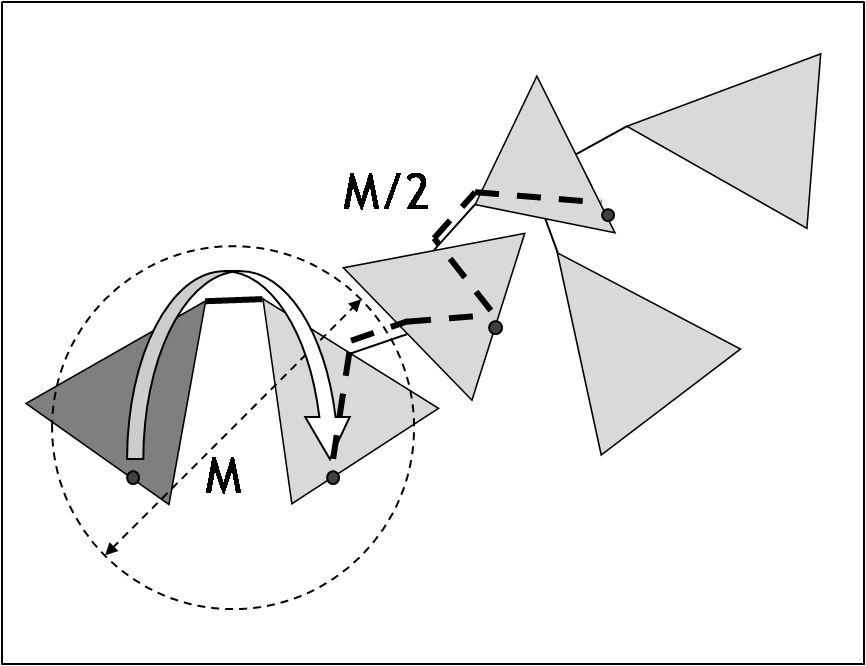}
\caption{Illustration of routine \extender. See also Figure~\ref{figure:extender}.}\label{fig:extend}
\end{center}
\end{figure}
\begin{figure*}
\framebox{
\begin{minipage}{11.7cm}
{\small \textbf{Algorithm} \extender\\
\textit{Input:} Component $\clusters{\radm}{i}$; 
Bipartitions $\{\minipart{j}{u}{v}\}_{j=1}^{\partsize{u}{v}}$; Leaves $u,v$;\\
\textit{Output:} Bipartitions $\{\fullpart{j}{u}{v}\}_{j=1}^{\partsize{u}{v}}$;
\begin{itemize}
\item For $j=1,\ldots,\partsize{u}{v}$ (unless $\partsize{u}{v}=0$):
\begin{itemize}
\item {\bf Initialization.} Denote by $\minipartone{j}{u}{v}$ the vertex set containing $u$ in the bipartition $\minipart{j}{u}{v}$,
and similarly for $v$; Initialize the extended partition $\fullpartone{j}{u}{v} := \minipartone{j}{u}{v}$,
$\fullparttwo{j}{u}{v} := \miniparttwo{j}{u}{v}$;
\item {\bf Modified Graph.} Let $K$ be $\clusters{\radm}{i}$ where all edges between
$\minipartone{j}{u}{v}$ and $\miniparttwo{j}{u}{v}$ have been removed;
\item {\bf Extension.} For all $w \in \vclusters{\radm}{i} - (\minipartone{j}{u}{v}\cup\miniparttwo{j}{u}{v})$,
add $w$ to the side of the partition it is connected to in
$K$ (by Proposition~\ref{proposition:outside}, each $w$ as above is connected to exactly
one side); 
\end{itemize}
\item Return the bipartitions $\{\fullpart{j}{u}{v}\}_{j=1}^{\partsize{u}{v}}$.
\end{itemize}
}
\end{minipage}
} \caption{Algorithm \extender. See Figure~\ref{fig:extend} for an illustration.} \label{figure:extender}
\end{figure*}

\section{Analysis of the Algorithm}\label{section:proof}

We assume throughout that $\distd$ is a $(\tau,M)$-distortion of $\dist$ and
moreover that $\radm$ satisfies the conditions of Theorem~\ref{theorem:main}.

\subsection{Leaf clustering: Determining the support of the forest}\label{section:clustering}

Recall the notation of Definition~\ref{definition:clustering}.
\begin{proposition}[Leaf Clustering]\label{proposition:clustering}
Let $\tau \leq M' \leq M-\tau$. 
Then
\begin{equation*}
\eclust{M'-\tau} \subseteq \ecluster{M'} \subseteq \eclust{M'+\tau}.
\end{equation*}
\end{proposition}
\begin{proof}
This follows immediately from the definition of $\distd$. Indeed,
if $\dist(u,v) < M'-\tau$ then
\begin{equation*}
\distd(u,v) < \dist(u,v) + \tau < (M'-\tau) + \tau < M'.
\end{equation*}
Similarly, if $\distd(u,v) < M'$ then
\begin{equation*}
\dist(u,v) < \distd(u,v) + \tau < M' + \tau.
\end{equation*}
\end{proof}

\subsection{Mini-reconstruction: Finding long edges on short paths}\label{section:mini}

Consider a component $\clusters{\radm}{i} = (\vclusters{\radm}{i}, \eclusters{\radm}{i})$ 
of $\cluster{\radm}$. Denote by
$\truetree{i} = (\truevertices{i}, \trueedges{i})$ the tree 
$T$ restricted to the leaves in $\vclusters{\radm}{i}$, that is,
\begin{itemize}
\item Keep only those edges of $T$ that are on paths between leaves in 
$\vclusters{\radm}{i}$;
\item Glue together edges adjacent to vertices of degree $2$; 
\item Equip $\truetree{i}$ with the metric $\dist$ restricted to 
$\vclusters{\radm}{i}\times \vclusters{\radm}{i}$ and
denote $\{\weight^{(i)}_e\}_{e\in \trueedges{i}}$ the corresponding weights.
\end{itemize}
\begin{proposition}[Chord Depth of $\truetree{i}$]\label{proposition:chord}
The chord depth of $\truetree{i}$ is less than $\radm + \tau$.
\end{proposition}
\begin{proof}
We argue by contradiction.
Let $e$ be an edge in $\truetree{i}$.
Suppose that the chord depth of $e$ in $\truetree{i}$ is $\geq \radm + \tau$.
Consider the bipartition $\{\psi^{(1)}, \psi^{(2)}\}$ defined by $e$ in $\truetree{i}$.
Then it follows that for all $u_1 \in \psi^{(1)}$ and $u_2 \in \psi^{(2)}$, we have
\begin{equation*}
\distd(u_1,u_2) > \dist(u_1,u_2) - \tau \geq \radm,
\end{equation*}
so that $\clusters{\radm}{i}$ cannot be connected, a contradiction.
\end{proof}

Let $e' = (u',v')$ be an edge in a tree $T'$ with leaf set $L'$. We denote by 
$L'_{u' \to v'}$ the leaves of $T'$ that can be reached from $v'$ without
going through $u'$. Recall that for two leaves $u',v'$ of $T'$, we denote by $\vpath_{T'}(u',v')$
the set of vertices on the path between $u'$ and $v'$ in $T'$.
Recall also that
\begin{equation*}
\ball{\bprime}{i}(u,v) = \left\{w \in \vclusters{\radm}{i}\ :\ 
\distd(u,w) \lor \distd(v,w) < M \right\},
\end{equation*}
for $u,v\in\vclusters{\radm}{i}$.
\begin{proposition}[Witnesses in $\ball{\bprime}{i}(u,v)$]\label{proposition:witness}
Assume that $2 \radm + 3\tau < M$.
Let $(u,v) \in \eclusters{\radm}{i}$.
Let $(x,y)$ be an edge of $\truetree{i}$ such that $x \in \vpath_{\truetree{i}}(u,v)$ but
$y \notin \vpath_{\truetree{i}}(u,v)$. Then we have
\begin{equation*}
\ball{\bprime}{i}(u,v)\cap \trueleaves{i}_{x \to y} \neq \emptyset,
\end{equation*}
where $\trueleaves{i}$ is the set of leaves of $\truetree{i}$.
\end{proposition}
\begin{proof}
By Proposition~\ref{proposition:chord}, there are leaves $x_0, y_0$ in
$\trueleaves{i}$ such that $(x,y) \in \path_{\truetree{i}}(x_0,y_0)$ and
$\dist(x_0,y_0) < \radm + \tau$. Assume without loss of generality that
$y_0 \in \trueleaves{i}_{x \to y}$. By assumption, 
\begin{equation*}
\dist(u,v) < \distd(u,v) + \tau < \radm + \tau.
\end{equation*}
Therefore,
\begin{equation*}
\dist(u,y_0) \leq \dist(u,x) + \dist(x,y_0) \leq \dist(u,v) + \dist(x_0,y_0) 
< 2\radm + 2\tau,
\end{equation*}
from which we get $\distd(u,y_0) < 2\radm + 3\tau < M$.
The same inequality holds for $\distd(v,y_0)$.
\end{proof}

Fix a pair of leaves $u,v$ with
$(u,v) \in \eclusters{\radm}{i}$. For $w\in \ball{\bprime}{i}(u,v)$, let 
\begin{equation*}
\intersect_w := \frac{1}{2}\left(
\distd(u,v) + \distd(u,w) - \distd(v,w)
\right),
\end{equation*}
and
\begin{equation*}
\trueintersect_w := \frac{1}{2}\left(
\dist(u,v) + \dist(u,w) - \dist(v,w)
\right).
\end{equation*}
Note that $\trueintersect_w$ is the distance between $u$ and
the intersection point of $\{u,v,w\}$.
Let $\{C_j\}_{j=0}^{\partsize{u}{v}}$ and $\{\minipart{j}{u}{v}\}_{j=1}^{\partsize{u}{v}}$ be as in
Figure~\ref{figure:minicontractor}. We write $w \sim w'$ if
$w,w' \in C_j$ for some $j$. Similarly, we write $w \lesssim w'$
(respectively $w < w'$)
if $w \in C_j$ and $w' \in C_{j'}$ with $j \leq j'$ (respectively $j < j'$).
\begin{proposition}[Intersection Points]\label{proposition:intersection}
Let $u,v$ be as above.
Then we have the following:
\begin{enumerate}
\item \label{item:1}\itemname{Identity}
If $x, y \in \ball{\bprime}{i}(u,v)$ are such that $\trueintersect_{x} = \trueintersect_{y}$ then
$x \sim y$;
\item \label{item:2}\itemname{Precedence}
If $x, y \in \ball{\bprime}{i}(u,v)$ are such that $\trueintersect_{x} \leq \trueintersect_{y}$ then
$x \lesssim y$;
\item \label{item:3}\itemname{Separation}
If $x, y \in \ball{\bprime}{i}(u,v)$ are such that $\trueintersect_{x} < \trueintersect_{y} - 4\tau$ 
and there is no $z \in \ball{\bprime}{i}(u,v)$ with $\trueintersect_{x} < \trueintersect_{z} < \trueintersect_{y}$, 
then $x < y$.
\end{enumerate}
\end{proposition}
\begin{proof}
For Part~\ref{item:1}, note that $\trueintersect_x = \trueintersect_y$ implies
\begin{equation*}
\left|\intersect_x - \intersect_y\right| < 2\tau.
\end{equation*}
(Note that the term $\distd(u,v)$ appears in both $\intersect_x$ and $\intersect_y$ and therefore does not contribute
to the error. The same argument applies to the error calculations below.)
Therefore, $x$ and $y$ are necessarily placed in the same $C_j$, that is,
$x \sim y$. See Figure~\ref{figure:minicontractor}.

For Part~\ref{item:2}, suppose by contradiction that $x > y$. Then
we have necessarily
\begin{equation*}
\intersect_x \geq \intersect_y + 2\tau,
\end{equation*}
which implies
\begin{equation*}
\trueintersect_y < \trueintersect_x - 2\tau + 2\tau \leq \trueintersect_x,
\end{equation*}
a contradiction.

For Part~\ref{item:3}, let
\begin{equation*}
X_0 = \{w\in \ball{\bprime}{i}(u,v)\ \mathrm{s.t.}\ \trueintersect_w \leq \trueintersect_x\},
\end{equation*}
\begin{equation*}
Y_0 = \{w\in \ball{\bprime}{i}(u,v)\ \mathrm{s.t.}\ \trueintersect_w \geq \trueintersect_y\},
\end{equation*}
\begin{equation*}
x_0 = \arg\max \{\intersect_w\ :\ w\in X_0\},
\end{equation*}
(breaking ties arbitrarily) and similarly
\begin{equation*}
y_0 = \arg\min \{\intersect_w\ :\ w\in Y_0\}.
\end{equation*}
Note that by assumption the pair $X_0, Y_0$ forms a partition of $\ball{\bprime}{i}(u,v)$.
By assumption,
\begin{equation*}
\trueintersect_{x_0} \leq \trueintersect_{x} < \trueintersect_{y} - 4\tau \leq \trueintersect_{y_0} - 4\tau,
\end{equation*}
which implies for all $x' \in X_0$ and $y' \in Y_0$
\begin{equation*}
\intersect_{y'} \geq \intersect_{y_0} > \intersect_{x_0} + 4\tau - 2\tau \geq \intersect_{x_0} + 2\tau \geq \intersect_{x'} + 2\tau.
\end{equation*}
Therefore, we have $x < y$.
\end{proof}

\begin{proposition}[Mini Reconstruction]\label{proposition:mini}
Let $u,v$ be as above.
Assume that $2 \radm + 3\tau < M$.
Then we have the following:
\begin{enumerate}
\item \label{item:11}\itemname{Reconstructed\ Edges\ Are\ Correct}
For each $j = 1,\ldots,\partsize{u}{v}$, there is a unique edge $e$ in $\trueedges{i}$
such that
\begin{equation*}
\bipart_{\truetree{i}}(e)\cap \ball{\bprime}{i}(u,v) = \minipart{j}{u}{v},
\end{equation*}
where the intersection on the left is applied separately to each
set in the partition;

\item \label{item:22}\itemname{Long\ Edges\ Are\ Present}
Let $e \in \trueedges{i}$ with $e\in \path_{\truetree{i}}(u,v)$ and $\weight^{(i)}_e > 4 \tau$.
Then there is a unique $j$ such that
\begin{equation*}
\bipart_{\truetree{i}}(e)\cap \ball{\bprime}{i}(u,v) = \minipart{j}{u}{v}.
\end{equation*}
\end{enumerate}
\end{proposition}
\begin{proof}
Part~\ref{item:11} follows from Proposition~\ref{proposition:witness} and
Proposition~\ref{proposition:intersection} Part~\ref{item:2}. Indeed, by 
Proposition~\ref{proposition:intersection} Part~\ref{item:2}, $\minipart{j}{u}{v}$
is a correct bipartition of $\truetree{i}$ restricted to $\ball{\bprime}{i}(u,v)$.
It corresponds to a unique edge of the latter tree because it is a full bipartition
of $\ball{\bprime}{i}(u,v)$. By Proposition~\ref{proposition:witness}, every edge
of $\truetree{i}$ is witnessed in $\ball{\bprime}{i}(u,v)$, so $\minipart{j}{u}{v}$
must also correspond to a unique edge in $\truetree{i}$.

Similarly, Part~\ref{item:22} follows from Proposition~\ref{proposition:witness}
and Proposition~\ref{proposition:intersection} Parts~\ref{item:2} and~\ref{item:3}.
\end{proof}

\subsection{Extending bipartitions: Reconstructing the components}\label{section:bipartitions}

Let $u, v \in \vclusters{\radm}{i}$ with $(u,v)\in \eclusters{\radm}{i}$
and let $\minipart{j}{u}{v}$ be one of the bipartitions returned by \minicontractor\ when
given $(\clusters{\radm}{i}; u, v)$ as input.
Let $e = (x,y) \in \trueedges{i}$ be the edge of $\truetree{i}$
corresponding to $\minipart{j}{u}{v}$ (as guaranteed by Proposition~\ref{proposition:mini})
and denote its bipartition by
\begin{equation*}
\bipart_{\truetree{i}}(e) = \{\bipart^{(u)},
\bipart^{(v)}\},
\end{equation*}
where $\bipart^{(u)}$ and $\bipart^{(v)}$ are respectively the sides containing
$u$ and $v$.
\begin{proposition}[Leaves Outside Ball]\label{proposition:outside}
Assume that $2\radm + 3\tau < M$.
Let $w \in \vclusters{\radm}{i} - \ball{\bprime}{i}(u,v)$.
Assume that $w \in \bipart^{(v)}$. 
Then, for all leaves $w'$ in $\bipart^{(u)}$ we have
\begin{equation*}
\distd(w,w') \geq \radm.
\end{equation*}
\end{proposition}
\begin{proof}
Assume by contradiction that there is $w' \in\bipart^{(u)}$
such that 
$\distd(w,w') < \radm$. 
The path between $w$ and $w'$ must go through $e$ since $w$ and $w'$
are on different sides of the partition. Therefore, for one of the endpoints of $e$,
say $x$, we have $\dist(w,x) < \radm + \tau$. Also, since
$\dist(u,x) \leq \dist(u,v) < \distd(u,v) + \tau < \radm + \tau$, we have
\begin{equation*}
\dist(w,u) < \dist(w,x) + \dist(x,u) < 2\radm + 2\tau < M.
\end{equation*}
We finally get
\begin{equation*}
\distd(w,u) < \dist(w,u) + \tau < 2\radm + 3\tau < M,
\end{equation*}
and similarly for $\distd(w,v)$, 
a contradiction since we assumed $w\notin \ball{\bprime}{i}(u,v)$.
\end{proof}
\begin{proposition}[Correct Extension]\label{proposition:extension}
The bipartition $\fullpart{j}{u}{v}$ returned by \extender\ is correct, that is,
$\fullpart{j}{u}{v} = \bipart_{\truetree{i}}(e)$.
\end{proposition}
\begin{proof}
Let $K$, $\minipartone{j}{u}{v}$, $\miniparttwo{j}{u}{v}$ be as in Figure~\ref{figure:extender}.
Since $\clusters{\radm}{i}$ is connected and we only remove edges between 
$\minipartone{j}{u}{v}$ and $\miniparttwo{j}{u}{v}$ to form $K$, 
it follows from Proposition~\ref{proposition:outside} that all vertices
in $\vclusters{\radm}{i} - (\minipartone{j}{u}{v}\cup\miniparttwo{j}{u}{v})$
are connected in $K$ to either $\minipartone{j}{u}{v}$ or $\miniparttwo{j}{u}{v}$.
\end{proof}
We finally get the following.
\begin{proposition}[Correctness of Main Loop]\label{proposition:mainloop}
Let $\{\esttree{i}\}_{i=1}^q$ be the trees obtained at the end
of the Main Loop of our algorithm. Then, for all $i=1,\ldots,q$, $\esttree{i}$ is a refinement of 
$F_{4 \tau, +\infty}(\truetree{i})$.
\end{proposition}
\begin{proof}
By Propositions~\ref{proposition:mini} and~\ref{proposition:extension},
all reconstructed edges are correct and they include at least those edges longer than $4 \tau$.
\end{proof}

\subsection{Path-disjointness: Length and depth of shared edges}\label{section:pruning}

Let $\truetree{i_1}$, $\truetree{i_2}$ be the tree $T$ restricted to
components $\clusters{\radm}{i_1}$, $\clusters{\radm}{i_2}$ respectively.
Note that each edge in $\truetree{i_j}$ is actually a path in $T$.
\begin{proposition}[Path-Disjointness]\label{proposition:shared}
For all $u_1, v_1 \in \trueleaves{i_1}$ and $u_2, v_2\in\trueleaves{i_2}$ such that
\begin{equation*}
\vpath_T(u_1,v_1) \cap \vpath_T(u_2,v_2) \neq \emptyset,
\end{equation*}
it holds that
\begin{enumerate}

\item \itemname{Depth\ of\ Shared\ Vertices} \label{item:222}
We have 
\begin{equation*}
\min\{\vertexdepth(z)\ :\ z\in\vpath_T(u_1,v_1) \cap \vpath_T(u_2,v_2)\} \geq \frac{1}{2}(\radm - 3\tau).
\end{equation*}

\item \itemname{Length\ of\ Shared\ Edges} \label{item:111} 
If, further, $\path_T(u_1,v_1) \cap \path_T(u_2,v_2) \neq \emptyset$ then
\begin{equation*}
\max\{\weight_e\ :\ e\in\path_T(u_1,v_1) \cap \path_T(u_2,v_2)\} \leq 2\tau.
\end{equation*}
\end{enumerate}
\end{proposition}
\begin{proof}
Let $z \in \vpath_T(u_1,v_1) \cap \vpath_T(u_2,v_2)$.
For $j=1,2$,
by Proposition~\ref{proposition:chord}, there are leaves $x_j, y_j$ in
$\trueleaves{i_j}$ such that $z \in \vpath_{\truetree{i_j}}(x_j,y_j)$ and 
$\dist(x_j,y_j) < \radm + \tau$.

For Part~\ref{item:222},
assume without loss of generality 
that $\dist(x_2,z) < \frac{1}{2}(\radm + \tau)$. Then,
for all $w\in \trueleaves{i_1}$,
\begin{eqnarray*}
\dist(w,z) 
&\geq& \dist(w, x_2) - \dist(z, x_2)\\
&\geq& \radm - \tau - \frac{1}{2}(\radm + \tau)\\
&\geq& \frac{1}{2}(\radm - 3\tau).
\end{eqnarray*}
A similar argument applies to $w \in \trueleaves{i_2}$ and $w \in L - (\trueleaves{i_1}\cup\trueleaves{i_2})$.

For Part~\ref{item:111}, let $e=(x,y) \in \path_T(u_1,v_1) \cap \path_T(u_2,v_2)$. 
Assume without loss of generality
that the path from $x$ to $y$ partitions $\{x_1,y_1,x_2,y_2\}$
as $\{\{x_1,x_2\},\{y_1,y_2\}\}$ in $T$, where $x_1,x_2,y_1,y_2$
were defined above.
We have
\begin{eqnarray*}
2 \dist(x,y) 
&=& \dist(x_1,y_1) + \dist(x_2,y_2) - \dist(x_1,x_2) - \dist(y_1,y_2)\\
&<& \distd(x_1,y_1) + \distd(x_2,y_2) - \distd(x_1,x_2) - \distd(y_1,y_2) + 4\tau\\
&<& 2\radm - 2\radm + 4\tau\\
&<& 4\tau,
\end{eqnarray*}
where the third line follows from the definition of the clustering graph
$\cluster{\radm}$.
\end{proof}

\subsection{Proof of Main Theorem}\label{section:proofmain}

\paragraph{Proof of Theorem~\ref{theorem:main}:}
Part~\ref{item:main1} follows from 
Proposition~\ref{proposition:shared}.
Recall that
\begin{equation*}
\radm < \frac{1}{2}[M - 3\tau].
\end{equation*}
Part~\ref{item:main2} then follows from Proposition~\ref{proposition:mainloop} and
Proposition~\ref{proposition:clustering}.
\qed

\section{Tightness of the Result}\label{section:lower}

We showed that given a $(\tau,M)$-distortion 
we reconstruct a subforest of $T$ with chord depth $\approx \frac{1}{2} M$ which includes
all edges of length at least $4\tau$. It may seem that we are losing a factor $2$ in the chord depth
and that, in fact, we should
be able to reconstruct edges of chord depth close to $M$. But this is not the case. We show in this section that
the chord depth of $\approx \frac{1}{2}M$ is essentially best possible (up to $O(\tau)$).

Consider the tree $T_0$ depicted in Figure~\ref{fig:lower1}. The tree $T_0$ has four leaves
$u, v, x_1, x_2$ with adjacent edges of length respectively
$4\tau$, $\frac{1}{2}M + 2\tau$, $\frac{1}{2}M + 4\tau$, and $\frac{1}{2}M + 4\tau$.
The middle edge has length $4\tau$ and the corresponding bipartition is $\{\{u,x_1\},\{v,x_2\}\}$.
\begin{figure}
\begin{center}
\includegraphics[width=7cm]{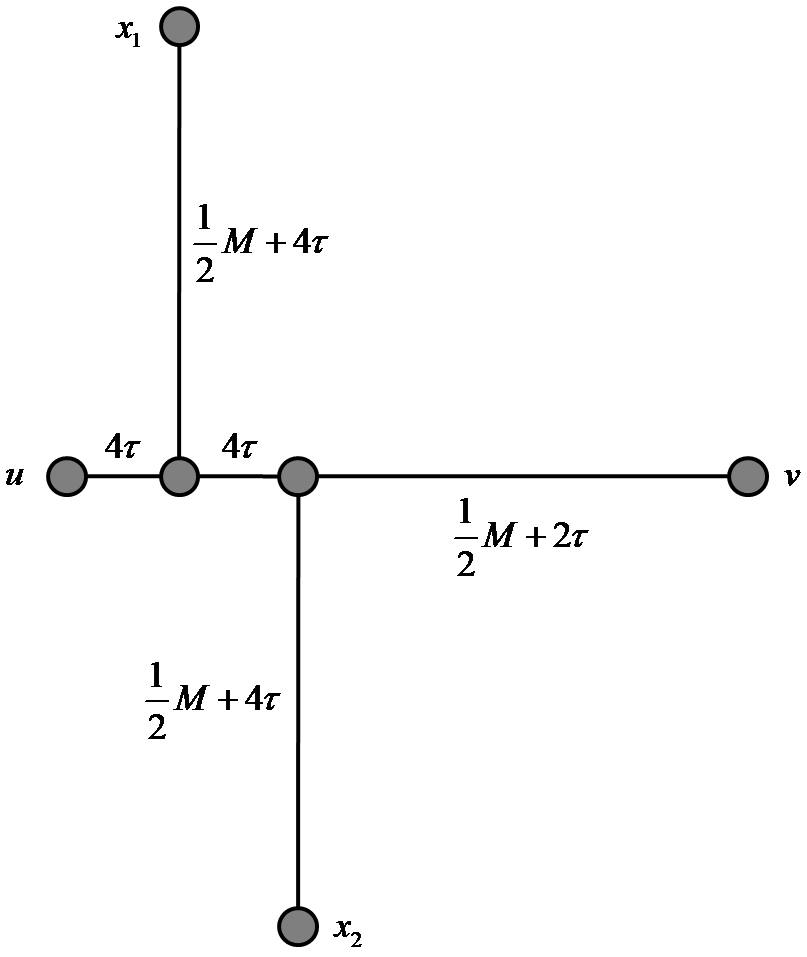}
\caption{Counter-example: Reference tree $T_0$.}\label{fig:lower1}
\end{center}
\end{figure}
Assume that we have the following $(\tau,M)$-distortion of the metric corresponding to $T_0$:
\begin{eqnarray*}
&&\distd_0(u,v) = \frac{1}{2}M + 10\tau,\ \distd_0(u,x_1) = \frac{1}{2}M + 8\tau,\ 
\distd_0(u,x_2) = \frac{1}{2}M + 12\tau,\\
&&\distd_0(v,x_1) = \distd_0(v,x_2) = \distd_0(x_1,x_2) = +\infty.
\end{eqnarray*}

Now, note that $\distd_0$ is also a $(\tau,M)$-distortion for the tree $T_1$
depicted in Figure~\ref{fig:lower2}. The tree $T_1$ has four leaves
$u, v, x_1, x_2$ with adjacent edges of length respectively
$4\tau$, $\frac{1}{2}M + 2\tau$, $\frac{1}{2}M$, and $\frac{1}{2}M + 8\tau$.
The middle edge has length $4\tau$ and the corresponding bipartition is $\{\{u,x_2\},\{v,x_1\}\}$.
\begin{figure}
\begin{center}
\includegraphics[width=7cm]{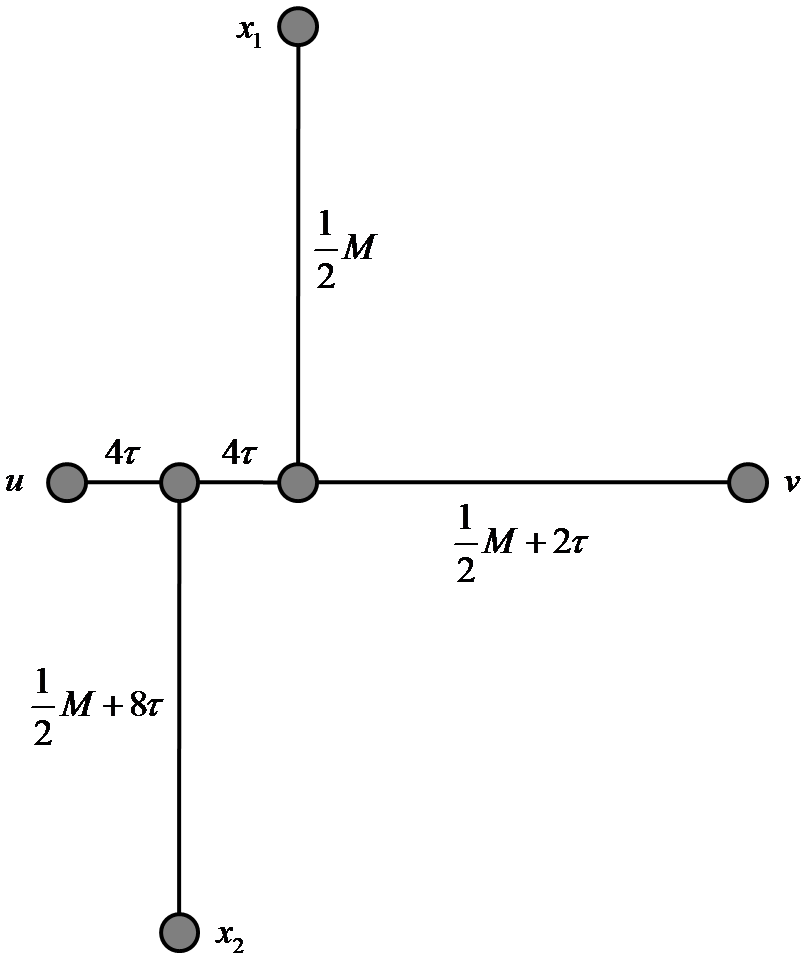}
\caption{Counter-example: Tree $T_1$ with equivalent distortion.}\label{fig:lower2}
\end{center}
\end{figure}

Hence, the two incompatible trees $T_0$ and $T_1$ cannot in general be distinguished from a $(\tau,M)$-distortion.
In particular, note that the middle edge of $T_0$ has length $4\tau$ and chord depth $\frac{1}{2}M + 10\tau$,
yet its bipartition cannot be recovered. This proves the claim.

\section{Implementation}\label{section:running}

We briefly discuss the running time of the algorithm.

Building the graph $\cluster{m}$ takes time $O(n^2)$, since we have to consider all pairs of leaves, and we find the connected components of $\cluster{m}$ with Breadth-First-Search in another $O(n^2)$. We argue next that, for $i=1,\ldots,q$, we need $O(n_i^5)$ to build $\esttree{i}$, where $n_{i}=|\vclusters{m}{i}|$. We show first that for all pairs of leaves $u$ and $v$, \minicontractor~and~\extender~take time $O(n_i^3)$. Indeed, \minicontractor~takes time $O(n_i)$, since its running time is linear in the size of $\ball{u}{v}(u,v)$; and \extender~takes time $O(n_i^3)$, since for each bipartition $\minipart{j}{u}{v}$---there are at most $O(n_i)$ of those---it is enough to perform a BFS. Given all bipartitions of the tree $\esttree{i}$, we use the standard TREE POPPING algorithm of~\cite{Meacham:81,BandeltDress:86} to build $\esttree{i}$; since we have $O(n_i^3)$ bipartitions (not all of them distinct) this last step takes time $O(n_i^4)$. So for each tree $i$ we need $O(n_i^5)$, and summing over $i$'s the total running time becomes $O(n^5)$.

We can improve on this running time by a more efficient implementation of \extender~as follows. For all $j=0,\ldots,r(u,v)$, we remove from the graph $\clusters{m}{i}$ all leaves in $\cup_{\ell \neq j} C_{\ell}$ and perform a Breadth-First-Search to discover the leaves $K_j \subset \vclusters{m}{i}\setminus \ball{u}{v}(u,v)$ reachable in $\clusters{m}{i}$ from the leaves in $C_j$. From an easy modification of Propositions~\ref{proposition:mini} ans~\ref{proposition:outside}, it follows that for every $w \notin \ball{u}{v}(u,v)$ there is at most one $j \in \{0,\ldots,r(u,v)\}$ such that $w$ is connected to a leaf in $C_j$. Given this, we can argue that we can recover the bipartitions $\fullpart{j}{u}{v}$, $j=1,\ldots,r(u,v)$ from the $K_j$'s. The overall time needed by the BFS's is $O(n_i^2)$, hence $\esttree{i}$ can be computed in time $O(n_i^4)$ and our total running time becomes $O(n^4)$.

The above implementation is wasteful in running a BFS for every pair of leaves $u$ and $v$ 
with the possibility of creating as many as $O(n^3)$ bipartitions, each requiring $O(n)$ 
storage. Note that there are in fact at most $n$ distinct bipartitions in $T$. 
To improve on the running time one may need to combine 
the BFS's performed in the above implementation by interleaving 
the \minicontractor~and~\extender~steps with the TREE POPPING algorithm.

\section{Concluding remarks}

An interesting question for future work is whether the {\em approximate}
disjointness in our results can be avoided. Since we guarantee
that any shared edge lies deep inside the forest, it is tempting
to simply remove all deep edges (say beyond $\radm/4$) from the output forest.
Unfortunately, many of these edges may in fact be contracted and moreover
they may be clustered in ``supernodes'' including both deep and not-so-deep
edges. It does not seem to be a trivial task to break these deep supernodes
apart and preserve strong reconstruction guarantees.  


\bibliographystyle{alpha}
\bibliography{thesis}{}

\clearpage

\appendix

%
%
\section{Log-Det Estimator}\label{section:markov}

For completeness, we relate the definition of the distorted metric (see Definition~\ref{def:distorted metric}) 
to its biological context. In phylogenetic reconstruction, a distorted metric is naturally derived 
from samples of a Markov model on a tree---a common model of DNA sequence evolution used in Biology.

\begin{definition}[Markov model on a tree]\label{def:mmt}
A {\em Markov model on a tree} is the following stochastic
process:
\begin{itemize}
\item Let $T_\rho=(V,E,\rho)$ be a finite tree rooted at $\rho$.
Denote by $E_\downarrow$ the set $E$ directed away from the root.
\item Let $L = [n]$ be the leaf set of $T_\rho$.
\item Let $\rcal$ be a finite set with $r$ elements.
\item Associate to each edge $e\in E$ a $r\times r$ stochastic matrix
$M(e)$ with $\det M(e) > 0$.
\item Let $\pi_\rho$ be a distribution on $\rcal$ with $\pi_\rho(\sigma) > 0$
for all $\sigma\in \rcal$.
\end{itemize}
The process runs as follows. Pick a state for the root according
to $\pi_\rho$. Moving away from the root toward the leaves,
apply the channel $M(e)$ to each edge $e$ independently. 
Denote the state so obtained $\sigma_V = (\sigma_v)_{v\in V}$. In particular,
$\sigma_{[n]}$ is the state at the leaves. 
More precisely, the joint distribution of $\sigma_V$ is given by
\begin{equation*}
\mu_V(\sigma_V) = \pi_{\rho}(\sigma_\rho) 
\prod_{e = (x,y) \in E_\downarrow} (M(e))_{\sigma_{x} \sigma_{y}},
\end{equation*}
and therefore the distribution at the leaves is
\begin{equation*}
\mu_L(\sigma_L) = \sum_{\sigma'_V : \sigma'_L = \sigma_L}
\pi_{\rho}(\sigma'_\rho) 
\prod_{e = (x,y) \in E_\downarrow} (M(e))_{\sigma'_{x} \sigma'_{y}}.
\end{equation*}
For $W \subseteq V$, we denote by $\mu_W$ the marginal
of $\mu_V$ at $W$. 
\end{definition}
More generally, we are given
$k$ independent samples $(\sigma^{i}_{[n]})_{i=1}^k$
from the same Markov model. 
We think of $(\sigma_l^i)_{i=1}^k$
as the sequence at $l \in [n]$. Typically in biological applications 
$\rcal = \{\mathrm{A}, \mathrm{G},\mathrm{C},\mathrm{T}\}$.
MMTs model how DNA sequences stochastically evolve by point mutations
along an evolutionary tree---under the assumption that each site in the sequences evolves 
independently.

In the phylogenetic reconstruction problem, we are given sequences 
$(\sigma^{i}_{[n]})_{i=1}^k$ (one sequence for each extant species) 
and our goal is to recover the generating tree---or more precisely
its unrooted version (the root is typically not identifiable~\cite{Steel:94}). 
A natural place to start is to measure a notion of
``distance'' between the leaves. 
That is, we seek to associate to an MMT
an additive metric as defined in Definition~\ref{def:logdet}. 
In general, this can be achieved using the so-called log-det distance.
\begin{definition}[Log-Det Distance~\cite{Steel:94}. See also~\cite{BarryHartigan:87,
LoStHePe:94,Lake:94}.]\label{def:logdet}
Consider the Markov model in Definition~\ref{def:mmt}.
Associate to each edge $e = (u,v)\in E_\downarrow$ a weight $\weight(e)$ as follows:
\begin{itemize}
\item If $e$ is a leaf edge then
\begin{equation*}
\weight(e) = -\log  \det M(e)  - \frac{1}{2} \log \prod_{\sigma' \in \rcal} \mu_{u}(\sigma').
\end{equation*}
\item Otherwise
\begin{equation*}
\weight(e) = -\log  \det M(e)  
- \frac{1}{2} \log \prod_{\sigma' \in \rcal} \mu_{u}(\sigma')
+ \frac{1}{2} \log \prod_{\sigma' \in \rcal} \mu_{v}(\sigma').
\end{equation*}
\end{itemize}
The {\em log-det distance} is defined as:
$\forall u,v \in L,$
\begin{equation*}
\dist(u,v)
\equiv -\log \det F(u,v)
= \sum_{e\in\path_T(u,v)} \weight_e,
\end{equation*}
where
\begin{equation*}
\forall \sigma',\sigma'' \in \rcal,\ (F (u,v))_{\sigma',\sigma''} 
= \mu_{\{u,v\}}(\sigma_u = \sigma',\sigma_v = \sigma'').
\end{equation*}
It was shown in~\cite{Steel:94} that the log-det distance is indeed
an additive metric.
\end{definition}
When the sequence length $k$ is finite, we can only obtain an estimate $\distd$ of $\dist$
\begin{equation*}
\distd(u,v)
= -\log \det \widehat{F}(u,v),
\end{equation*}
where
\begin{equation*}
\forall \sigma',\sigma'' \in \rcal,\ (\widehat{F}(u,v))_{\sigma',\sigma''} 
= \frac{1}{k}\sum_{i=1}^k\ind\{\sigma^i_u = \sigma',\sigma^i_v = \sigma''\}.
\end{equation*}
The next lemma,
a slight generalization of Proposition 2.1 in~\cite{Mossel:07},
shows that such an estimator constitutes a {\em distorted metric}.

\begin{lemma}[Log-Det Distance: Distorted Metric]\label{prop:logdet}
Let $\distd$ be the estimator defined above. Then there is a constant $\Lambda > 0$ such that
if one chooses $(\tau, M)$ with
\begin{equation*}
k \geq \frac{\Lambda}{(1 - e^{-\tau})^2 } e^{2M + 4\tau} \log n,
\end{equation*}  
then $\distd$ is a $(\tau, M)$-distortion with probability
$1 - 1/\poly(n)$.
\end{lemma}

\begin{proof}
Fix $u,v \in L$. Denote 
$F = F(u,v)$, 
$\widehat F = \widehat F(u,v)$,
$\omega = \dist(u,v)$, and
$\hat\omega = \distd(u,v)$.
We assume that $k$ is at least $\Omega(\log n)$.
Let $\widehat F'$ be $\widehat F$ with
one sample arbitrarily changed.
It was argued in~\cite{ErStSzWa:99b} that
there are constants $c_1, c_2$ such that
\begin{equation*}
\left|\det \widehat F - \det \widehat F'\right| \leq \frac{c_1}{k},
\end{equation*}
and
\begin{equation*}
\left|\det F - \expec[ \det \widehat F ]\right| \leq \frac{c_2}{k}.
\end{equation*}

Assume $\omega < M + 2\tau$
(in particular if $\omega < M + \tau$). 
By Azuma's inequality,
\begin{eqnarray*}
\prob[\hat\omega \geq \omega + \tau]
&=& \prob[\det \widehat F - \det F \leq - (\det F) (1 - e^{-\tau})]\\
&\leq& \prob\left[\det \widehat F - \expec[\det \widehat F]
\leq e^{-M-2\tau}(1 - e^{-\tau}) - \frac{c_2}{k}\right]\\
&\leq& \exp\left(-\frac{k}{2 c_1^2} \left(
e^{-M-2\tau}(1 - e^{-\tau}) - \frac{c_2}{k}
\right)^2\right),
\end{eqnarray*}
where we assume $k$ is large enough that
\begin{equation*}
e^{-M-2\tau}(1 - e^{-\tau}) - \frac{c_2}{k} \geq 0.
\end{equation*}
The same inequality holds for 
$\prob[\hat\omega \leq \omega - \tau]$.

On the other hand,
assume $\omega > M + 2\tau$. Then,
\begin{eqnarray*}
\prob[\hat\omega \leq M + \tau]
&\leq& \prob\left[\det \widehat F - \expec[\det \widehat F]
\geq e^{-M-\tau} - e^{-M - 2 \tau} - \frac{c_2}{k}\right]\\
&\leq& \exp\left(-\frac{k}{2 c_1^2} \left(
e^{-M-\tau}(1 - e^{-\tau}) - \frac{c_2}{k}
\right)^2\right).
\end{eqnarray*}
\end{proof}

\end{document}